\documentclass[11pt, a4paper, oneside]{article}

\usepackage[letterpaper, left=2.3cm, right=2.3cm, top=2cm, bottom=2cm]{geometry}%
\usepackage[utf8]{inputenc}
\usepackage{graphicx}
\usepackage{amssymb}
\usepackage{amsmath}
\usepackage{amsfonts}
\usepackage{framed}
\usepackage[english]{babel}
\usepackage{forloop}
\usepackage{etoolbox}
\usepackage{wasysym}
\usepackage{array}
\usepackage{amsthm}
\usepackage{enumitem}
\setlist[description]{leftmargin=\parindent,labelindent=\parindent}

\usepackage{comment}
\usepackage{nicefrac}
\usepackage{algorithm,algpseudocode}
\usepackage[disable]{todonotes}
\newcommand{\dtodo}[2][]{\todo[color={red!100!green!33}, #1]{#2}}

\usepackage{xspace}
\usepackage{mathtools}
 
\usepackage[sort,numbers]{natbib}
\newcommand{\?}[1]{
  \sbox0{A#1}\sbox2{A\kern0pt #1}%
  \kern\dimexpr\wd0-\wd2\relax
  #1%
}

\usepackage{tikz}
\usetikzlibrary{shapes.misc,shapes}
\usetikzlibrary{decorations.pathreplacing}
\usetikzlibrary{automata}
\usetikzlibrary{petri}
\usetikzlibrary{arrows}
\usetikzlibrary{calc,positioning,shadows.blur}

\newtheorem{theorem}{Theorem}
\newtheorem{definition}{Definition}
\newtheorem{lemma}{Lemma}

\newcommand{\kdtau}{\tau}

\usepackage{newfloat}
\DeclareFloatingEnvironment{algo}
 
\algdef{SE}[SUBALG]{Indent}{EndIndent}{}{\algorithmicend\ }%
\algtext*{Indent}
\algtext*{EndIndent}

\tikzset{cross/.style={cross out, thick, draw=black, minimum size=2*(#1-\pgflinewidth), inner sep=0pt, outer sep=0pt},
cross/.default={5pt}}

\title{Byzantine Preferential Voting}
\author{
\begin{tabular}{p{0.3\textwidth}cp{0.3\textwidth}}
\raggedleft Darya Melnyk & Yuyi Wang &   Roger Wattenhofer \\
\raggedleft \small\texttt{dmelnyk@ethz.ch} & \small \texttt{yuwang@ethz.ch} &  \small \texttt{wattenhofer@ethz.ch}\\
\\
& \small  ETH Z\"urich, Switzerland
\end{tabular}
}
\date{}

\begin{document}
 
\begin{titlepage}
\maketitle
\thispagestyle{empty}
\begin{abstract}
In the Byzantine agreement problem, $n$ nodes with possibly different input values aim to reach agreement on a common value in the presence of $t<n/3$ Byzantine nodes which represent arbitrary failures in the system. This paper introduces a generalization of Byzantine agreement, where the input values of the nodes are preference rankings of three or more candidates. We show that consensus on preferences, which is an important question in social choice theory, complements already known results from Byzantine agreement. In addition preferential voting raises new questions about how to approximate consensus vectors. We propose a deterministic algorithm to solve Byzantine agreement on rankings under a generalized validity condition, which we call Pareto\,-\?Validity. These results are then extended by considering a special voting rule which chooses the Kemeny median as the consensus vector. For this rule, we derive a lower bound on the approximation ratio of the Kemeny median that can be guaranteed by any deterministic algorithm. We then provide an algorithm matching this lower bound. To our knowledge, this is the first non-trivial multi-dimensional approach which can tolerate a constant fraction of Byzantine nodes.
\vspace{1cm}
\dtodo[inline]{Kemeny rule and Kemeny median. NOT: rank aggregation, Kemeny rank}
\dtodo[inline]{we use $t$, not $f$}
\dtodo[inline]{all Validities using \textbackslash?}
\dtodo[inline]{there is no weakly paretian validity anymore}
\dtodo[inline]{All examples with rankings use either $c_1,c_2,\ldots,c_m$ or $c_i,c_j,c_k$}
\dtodo[inline]{there should be no alternatives, only candidates (alternatives are used once in the motivation)}

\end{abstract}
\end{titlepage}


\section{Introduction} 
In the first-past-the-post (FPTP) voting mechanism, each voter indicates his/her candidate of choice on a ballot, and the candidate that received the most votes wins. This is a popular voting mechanism that for instance is used in the United States. FPTP voting is related to consensus or Byzantine agreement in the sense that each voter/node has an input, and the voters need to decide on a single output. If Byzantine voters are present, it is natural to try to agree on a candidate with many votes to be robust against Byzantine behavior, e.g., \cite{BenOr,BrachaRB,QueenAlgorithm}. 
FPTP does however face a lot of criticism since it has several issues, in particular the problem of wasted votes to minority parties, but also tactical voting or gerrymandering \cite{orvis2013introducing}. Moreover, one may argue that FPTP will eventually lead to a two-party system, e.g., \cite{sachs2011price}. 

Preferential voting is a powerful alternative: Each voter \textit{ranks} the candidates first, second, third, etc. The collection of all rankings forms a \textit{preference profile} from which a winner (or even a ranking) is determined. Preferential voting is more expressive, harder to manipulate, and solves many of FPTP's problems. In particular there is no problem with wasted votes to minority candidates. If inputs are not just binary, preferential voting will lead to much better decisions. It is therefore remarkable that byzantine agreement research has not given preferential voting any attention.


In this paper we want to investigate how robust preferential voting is in a Byzantine environment. 
In Section \ref{sec:motivation}, we first focus on some basic properties for voting rules, and see that not all of them can be satisfied if the nodes should reach agreement. This is because Byzantine voters are manipulators that modify the result to make it more favorable to themselves. In the main part of the paper (Section \ref{sec:kemeny}) we study how well the voting result intended by the correct (non-Byzantine) voters can be approximated. For this purpose we introduce the Kemeny rule which picks the most central ranking as the voting result. We will provide an algorithm that approximates the solution of the Kemeny rule in the presence of Byzantine voters and prove that this algorithm computes the best possible approximation. We believe that our paper will help to get a deeper understanding of both fault-tolerant distributed systems as well as social choice theory.

\section{Background and Motivation}\label{sec:motivation}
In search of a \emph{fair} rule to elect candidates, philosophers and mathematicians started developing various voting mechanisms and rules already in the beginning of the 18th century. In the middle of the 20th century, Kenneth Arrow \cite{Arrow1951,Arrow1963} was one of the first to formalize existing rules and analyze possibility and impossibility results in an axiomatic fashion, thereby introducing the field of Computational Social Choice. In this section we use this formalism in order to show how well Byzantine agreement connects to voting theory. 

We start by considering the special case of $n$ voters voting on only two candidates $c_1$ and $c_2$. 
In this setting, each voter (node) ranks the two candidates such that its preferred candidate (input value) is ranked first. A vote for a candidate $c_1$ means that the voter strictly prefers $c_1$ to $c_2$, here denoted $c_1 \succ c_2$. A central authority then applies a \textit{social choice function (SCF)} to a given preference profile in order to determine the winner (decision value), or set of winners in case of a tie. An SCF $f$ can be qualified based on the following properties:

\begin{itemize}[noitemsep]
\item $f$ is \textit{anonymous} if interchanging two \textit{voters} (swapping their names) does not change the result 
\item $f$ is \textit{neutral} if renaming the \textit{candidates} (changing their names) does not change the result 
\item $f$ is \textit{positively responsive} if in a case where the decision is a tie ($c_1$ is among the winners) and a voter changes its ranking from $c_2 \succ c_1$ to $c_1 \succ c_2$, candidate $c_1$ becomes the unique winner
\end{itemize}

\noindent One example of an SCF is the \textit{majority rule}. It chooses the candidate that wins most pairwise comparisons against every other candidate. Note that such a winner always exists in elections with two candidates, but not necessarily in the general case with any number of candidates.
Social choice theory shows that the majority rule satisfies all desirable properties for the special case of voting on two candidates:

\begin{theorem}[May's Theorem \cite{May}]
For two candidates and any number of voters, the majority rule is the unique SCF that satisfies anonymity, neutrality and positive responsiveness.
\end{theorem}

Interestingly, most known algorithms for binary Byzantine agreement indirectly exploit the properties of May's theorem.
Some of them make use of leaders who suggest their decision value to all nodes \cite{KingAlgorithm, QueenAlgorithm}. 
The leader in these algorithms temporarily plays what is known as a dictator in voting theory. 
Another type of algorithm, e.g., the shared coin algorithm in \cite{BrachaRB}, is biased towards one of the outcomes and thus violates neutrality.
In general we can say that most of the proposed algorithms try to use the majority value as the decision value if a majority exists, or an arbitrary input value otherwise, see for example \cite{BenOr,BrachaRB}. Such settings may satisfy anonymity and neutrality, but in cases where the correct nodes are undecided, i.e. there is a tie between the two input values, Byzantine nodes have a large influence on the majority value. Thus, if a correct node decides to swap two candidates in its ranking in order to make one of the candidates win, a Byzantine node can perform an opposite swap in its own ranking and return the profile to the previous state. This shows that positive responsiveness cannot be satisfied for these algorithms in the presence of Byzantine nodes.

May's theorem does not apply to the general case with more than two candidates. In fact, the majority rule gives surprisingly bad results for three or more candidates. To illustrate this, let $m$ denote the number of candidates. Assume that $n/2+1$ voters rank the candidates as $c_1 \succ c_2 \ldots \succ c_{m}$, and all other voters rank the candidates as $c_2 \succ c_3 \succ \ldots \succ c_1$. In this case candidate $c_1$ wins every pairwise comparison according to the majority rule, even though $c_2$ seems to be the candidate that is approved by more voters.

Moreover, a lot of information is lost when a single winner is sought. When it comes to preferential voting, social choice theory therefore often wants not only the input to be rankings but also the output. More formally:

\begin{definition}[Social Welfare Function]
A Social Welfare Function (SWF) is a map from a preference profile to a set of consensus rankings.
\end{definition}

\noindent For an SWF $g$, the following three properties are usually considered:

\begin{itemize}[noitemsep]
\item $g$ is \textit{dictatorial} if there is one distinguished voter whose input ranking is chosen as the single consensus ranking
\item $g$ is \textit{independent of irrelevant alternatives (IIA)} if the consensus ranking of two candidates $c_i$ and $c_j$ only depends on the relative preference of these candidates in each voter's ranking, and not on the ranking of some third candidate $c_k$
\item $g$ is \textit{weakly Paretian} if it satisfies the weak Pareto  condition \cite{pareto1919}: for two candidates $c_i$ and $c_j$ which are ranked $c_i \succ c_j$ by all voters, consensus ranking has to rank $c_i \succ c_j$ as well
\end{itemize}

\noindent In contrast to IIA and weak Pareto, dictatorship is a highly undesirable property in voting theory. Unfortunately, a corresponding result to May's theorem for SWF's on three or more candidates is the famous impossibility result by Arrow:

\begin{theorem}[Arrow's Impossibility Theorem \cite{Arrow1951}]\label{thm:Arrow}
If there are at least three candidates which the members of the society are free to order in any way, then every SWF that is weakly Paretian and IIA must be dictatorial.
\end{theorem}

From the viewpoint of Byzantine agreement, an SWF should not be dictatorial since one does not want a dictator 
to be a Byzantine node and choosing more than one dictator may also result in different decision values. Consequently, any reasonable Byzantine agreement protocol must either violate IIA or weak Pareto. 
The IIA condition implies that the consensus ranking should remain the same if the input of every correct node does not change, no matter what the Byzantine nodes do. 
However, a Byzantine node can pretend to be a correct node but change its ranking in two executions in which the correct nodes have the same inputs. This change may lead to a different consensus ranking which would violate IIA.
For the weak Pareto condition consider the case with two candidates: if every non-Byzantine voter ranks $c_1 \succ c_2$, the consensus ranking should also rank $c_1 \succ c_2$. This corresponds to a well-known validity condition in Byzantine agreement -- the \textit{All\hspace{0.1em}-\hspace{0.05em}Same\hspace{0.1em}-\?Validity}: If all correct nodes have the same input value, all correct nodes have to decide on this value. We use the weak Pareto condition to impose a validity rule on Byzantine Agreement with rankings:


\begin{description}
\item[Pareto\,-\?Validity] for any pair of candidates $c_i$ and $c_j$: if all correct nodes rank $c_i\succ c_j$, then the consensus ranking should rank $c_i\succ c_j$ as well.
\end{description}

Given $m$ candidates, Pareto\,-\?Validity can be viewed as All\hspace{0.1em}-\hspace{0.05em}Same\hspace{0.1em}-\?Validity applied on each of the $\binom{m}{2}$ pairs of candidates in a ranking. Note that Byzantine agreement on a ranking is at least as hard as binary Byzantine agreement: Consider a case where the nodes agree on the ranking of the candidates $c_{3}, \ldots c_{m}$ which they rank last, but not on the two first candidates $c_1$ and $c_2$. The Pareto condition is then satisfied for every binary relation which contains at least one of the candidates $c_{3}, \ldots c_{m}$. Agreement in this case is reduced to binary Byzantine agreement on the two candidates $c_1$ and $c_2$, under the All\hspace{0.1em}-\hspace{0.05em}Same\hspace{0.1em}-\?Validity condition.

Unfortunately, there is no straightforward way to apply a binary Byzantine agreement protocol to solve Byzantine agreement on rankings. Other than binary relations on two candidates, preference profiles can form cycles, e.g., they can contain all three relations $c_i\succ c_j$, $c_j\succ c_k$ and $c_k\succ c_i$ which are each preferred by a majority of nodes. 
The smallest preference profile which produces such a cycle of binary relations is called a \textit{Condorcet cycle}. It contains three rankings $c_i\succ c_j\succ c_k$, $c_j\succ c_k\succ c_i$ and $c_k\succ c_i\succ c_j$ which induce the three relations from above. Simply agreeing on each pair of candidates can thus lead to a circular decision which does not form a ranking. In order to get rid of cycles one could think of applying the quicksort algorithm on the candidates sorted with respect to the majority. This procedure however violates Pareto\,-\?Validity: Consider a candidate $c_i$ that Pareto dominates candidate $c_j$. Assume that the quicksort algorithm compares both candidates to some third candidate $c_k$ first. Then $c_j$ might win against $c_k$ and $c_i$ might lose, thus swapping $c_i$ and $c_j$ in the consensus ranking. 
This consideration makes the problem of finding a consensus ranking in the presence of Byzantine nodes rather an instance of multi-valued agreement, as we discuss in Section \ref{sec:king}, which makes the problem both interesting and challenging. 


\section{Related Work}

Byzantine agreement was first proposed as the Byzantine Generals problem by Pease, Shostak and Lamport \cite{ByzantineGeneralsPease,ByzantineGeneralsLamport}. 
In these papers the authors showed that three nodes cannot establish agreement in the presence of one Byzantine node even if the communication system is synchronous. Given $n$ nodes, it was shown for the synchronous model that at least $t+1$ rounds are required to establish agreement \cite{FischerLynchMinRounds}, where $t<n/3$ is the number of Byzantine nodes in the system; the corresponding upper bound was provided in \cite{KingAlgorithm,QueenAlgorithm}. For the asynchronous model, the FLP impossibility result \cite{FLPimpossibility} states that there is no deterministic agreement protocol which can tolerate even one Byzantine node. The first randomized algorithm for solving Byzantine agreement proposed in \cite{BenOr} had expected exponential running time for a constant fraction of Byzantine nodes. This result was recently improved in \cite{KingSaia2016}, where the authors showed that it is possible to establish agreement within expected polynomial running time using spectral methods. 

Byzantine agreement with more than two input values has mostly been considered in approximate agreement \cite{ApproximateAgreement,ApproximateAgreement2}, where the input values of the nodes converge towards some value  over rounds. More recent results seek to establish agreement on a value that makes sense for applications. In \cite{PowerOfTwoChoices}, the values converge towards a value at most $\sqrt{n\log{n}}$ positions away from the median. In \cite{MedianValidity,IntervalValidity} an exact algorithm to establish agreement on a value that is at most $t/2$ positions away from the median or $t$ positions away from  a minimum or a maximum was proposed. In \cite{VectorConsensus,VectorConsensusAsynch,VectorConsensusJointWork}, Byzantine agreement was further generalized to several dimensions and the nodes converge to a vector inside the convex hull of all correct input vectors. While the one-dimensional case has been investigated in depth, all previous approaches for multiple dimensions struggle to derive an algorithm which either can tolerate a constant fraction of Byzantine nodes independent on the number of dimensions, or find a solution that is not trivial.

In social choice theory, Byzantine behavior can be interpreted as manipulation of a ballot in an election, in which the manipulating party has full knowledge about all votes. Bartholdi et al.\ \cite{BartholdiManipulation} defined manipulation as a preference profile where one single voter can change its ranking such that this voter's most preferred candidate wins the election. Groups of voters have also been considered in this context, but mostly from the perspective of how hard it is for a group of nodes to manipulate the voting result given a certain voting rule \cite{manipulationHard1,manipulationHard2}. Other types of Byzantine behavior have been considered with respect to robustness of proposed voting rules. In \cite{robustVoting}, the authors investigate robustness of Borda's mean and median in the presence of outlier ballots. In \cite{robustVotingNoise}, robustness of scoring rules is considered under arbitrary noise which is described in terms of pairwise swaps of candidates in the ranking of one voter. 

In this paper we will consider the Kemeny rule which was first proposed in \cite{Kemeny1959,KemenySnell}. The corresponding Kemeny median satisfies additional properties to those presented in Section \ref{sec:motivation}, but it was shown to be NP-hard to compute for an increasing number of candidates and already for four voters in \cite{BartholdiNPhard,DworkNPhard}. 
At least three different $2$-approximation algorithms for the Kemeny median have been proposed in \cite{Ailon} and \cite{DiaconisGraham} respectively. In \citep{Ailon}, the approximation ratio was improved to $4/3$ using randomization, and later derandomized in \cite{vanZuylen2008}. A good overview over the Kemeny rule and an extended introduction into social choice theory can be found in \cite{Brandt2016}.


\section{A Deterministic Algorithm for Pareto\,-\?Validity} \label{sec:king}

This section focuses on Byzantine agreement protocols for rankings that satisfy Pareto\,-\?Validity. By using single transferable voting and a multi-valued Byzantine agreement algorithm, a ranking satisfying Pareto\,-\?Validity can be obtained in $(m-1)\cdot (t+1)$ rounds: In the first $t+1$ rounds, we let the voters apply the King algorithm in order to agree on the top candidate. Then every node removes this candidate from its ranking. In the next step, they will agree on the top candidate from the reduced rankings, and so on. While this procedure is simple, the number of rounds depends not only on the number of nodes, but also on the number of candidates.

In the following we present a deterministic algorithm which solves this problem in only $t+1$ rounds using the same number of messages. We do this by modifying the King algorithm to broadcast rankings instead of single candidates. In the proposed algorithm we select $t+1$ different nodes and assign each of them to one of the $t+1$ rounds of the algorithm. Such a node is called a dictator of the corresponding round. This dictator then suggest its own, possibly adjusted, ranking to all nodes, which will always be accepted if the dictator is a correct node. This way, dictators decide on the ranking of all pairs of candidates which do not satisfy the Pareto\,-\?Validity. Algorithm \ref{alg:king} presents this procedure in pseudocode.

\algblockdefx{MyIf}{EndMyIf}[1]{\textbf{if} #1 \textbf{fix $\mathbf{c_k \succ c_l}$}}{\textbf{end if}}

\begin{algorithm}
               \begin{algorithmic}[1]
                   \Statex Every node $v$ executes the following algorithm
        \For{ round $1$ to $t+1$}
                             \Statex \hspace{\algorithmicindent}\textit{Communication Phase:}
                             \Indent
                \State Broadcast own input ranking $r_v$
                \For{all pairs of candidates $c_i$ and $c_j$}
                                 \If{$c_i$ is ranked above $c_j$ in at least $n-t$ rankings}\label{step:propose}
                        \State Broadcast ``propose $c_i \succ c_j$''
                    \EndIf
                \EndFor
                \MyIf{some ``propose $c_k\succ c_l$'' received at least $t+1$ times}\label{step:adjustRanking}
                    \State Adjust own ranking $r_v$ such that $c_k$ appears before $c_l$
                \EndMyIf
                             \EndIndent
                             \Statex \hspace{\algorithmicindent}\textit{Dictator Phase:}
            \Indent
                \State Let node $w$ be the predefined dictator of the current round
                \State The dictator broadcasts its ranking $r_{dictator} \coloneqq r_w$
            \EndIndent
            \Statex \hspace{\algorithmicindent}\textit{Decision Phase:}
            \Indent
                \If{$r_{dictator}$ agrees with $r_v$ in all fixed pairs $c_i \succ c_j$ from step \ref{step:adjustRanking}}\label{step:adaptDictator}
                   \State $r_v \coloneqq r_{dictator}$
                \EndIf
            \EndIndent
        \EndFor
        \State Return $r_v$
               \end{algorithmic}
               \caption{Byzantine agreement protocol on rankings (for $t < n/3$)}
               \label{alg:king}
\end{algorithm}

Since we are dealing with rankings, it is not trivial to see that the nodes will always be able to agree on a proper ranking at the end of the algorithm. In the following lemmas we will prove that the nodes can adjust their rankings in Step \ref{step:adjustRanking} of Algorithm \ref{alg:king} in order to guarantee Pareto\,-\?Validity and that the outcome of the algorithm will be a proper ranking. It is easy to see that the algorithm is correct for $t < n/4$ Byzantine nodes, since the correct nodes will not be able to propose binary relations which form a Condorcet cycle in this case.
In order to show that the algorithm can tolerate $t < n/3$ Byzantine nodes, we need to exploit the fact that no Byzantine node can propose relations that form a Condorcet cycle at any point of the algorithm.

\begin{lemma}\label{lemma:noCondorcet}
There is no Condorcet cycle that can be proposed by the correct nodes if $t < n/3$.
\end{lemma}
\begin{proof}
Assume by means of contradiction that the three relations $c_i\succ c_j$, $c_j\succ c_k$ and $c_k\succ c_i$ were each proposed by at least $t+1$ nodes in Step \ref{step:propose} of Algorithm \ref{alg:king}. Each binary relation was proposed by at least one correct node who must have seen $n-t$ nodes having a ranking with such a pair.

Let $t_1$ be the number of all Byzantine nodes who proposed $c_i \succ c_j$, $t_2$ the number of those who proposed $c_j \succ c_k$ and $t_3$ those nodes who proposed $c_k \succ c_i$. Further, let $t_{1\cap 2}$ denote the number of Byzantine nodes who proposed $c_i \succ c_j \succ c_k$. The following inequality then holds: $t_1 + t_2 -t_{1\cap 2} \leq t$.

The number of correct nodes who proposed $c_i \succ c_j \succ c_k$ is then $(n-t-t_1)+(n-t-t_2)+t_{1\cap 2} - (n-t) = n-t-t_1-t_2+t_{1\cap 2}$. The number of correct nodes who proposed $c_k\succ c_i$ is $n-t-t_3\geq n-2t$. However, the two sets must have a nonempty intersection, since $$n-t-t_1-t_2+t_{1\cap 2} + n-2t -(n-t) = n-2t -t_1-t_2+t_{1\cap 2} \geq n-3t > 1.$$ Therefore, there must be at  least one correct node who proposed $c_i\succ c_j\succ c_k$ and $c_k\succ c_i$ simultaneously. This is a contradiction.
\end{proof}

Note that by the properties of the King algorithm, no two opposite binary relations can be proposed in Step \ref{step:propose} simultaneously. Lemma \ref{lemma:noCondorcet} additionally shows that a Condorcet cycle cannot be proposed in Step \ref{step:propose} and thus all proposed pairs can form a ranking. It remains to show that the nodes will always be able to adjust their rankings to incorporate the proposed pairs.

\begin{lemma}\label{lem:adjustRanking}
In Step \ref{step:adjustRanking} a correct node will always be able to incorporate the proposed pairs into its own ranking.
\end{lemma}
\begin{proof}
This can be achieved by the following strategy: Divide the candidates into two sets. The first set contains all candidates which are in at least one of the pairs proposed by the $t+1$ nodes in Step \ref{step:adjustRanking}. This set of nodes will be ranked first. The second set will contain all candidates for which the node has not received any propose message. The candidates will be ranked second and will be dominated by all candidates from the first set. Next, we can rank all candidates in the first set according to the proposed relations, possibly leaving some pairs of the candidates not ranked. In the last step, all candidates which have not been ranked in each of the sets can be ranked by choosing binary relations from the local ranking of the node. This strategy outputs a ranking of candidates in which all proposed binary relations are satisfied.
\end{proof}

The next theorem summarizes the correctness results of Algorithm \ref{alg:king} and states that the consensus ranking will be valid, which can be derived with help of previous lemmas. The corresponding proof can be found in Appendix \ref{app:kingAlgo}.

\begin{theorem}\label{thm:kingAlgo}
At the end of Algorithm \ref{alg:king} all nodes will have agreed on the same ranking which additionally satisfies Pareto\,-\?Validity.
\end{theorem}


\section{Kemeny Median with Byzantine Nodes}\label{sec:kemeny}

Weakly Paretian voting rules are often not sufficient to pick the most fair ranking from a set of individual preference rankings. In search of the best possible consensus ranking we have to add restrictions on the voting rules without violating the known impossibility results of \cite{Arrow1951}. This leads us to majoritarian SWFs, one of which is the Kemeny rule.
In the following we will introduce the Kemeny rule and use it to derive a better consensus ranking in the presence of Byzantine nodes. Since Byzantine nodes have influence on the final ranking, the corresponding solutions can be qualified with respect to their approximation ratio which we define in Section \ref{sec:approx}. In Section \ref{sec:kemenyLB}, we will derive a lower bound on the approximation ratio of the Kemeny median in the presence of Byzantine nodes and further provide a matching upper bound in Section \ref{sec:kemenyAlgo}.

\begin{definition}[Kendall's $\kdtau$ distance \cite{Kendall}]
The Kendall's $\kdtau$ distance measures the distance between two rankings $r$ and $p$ on candidates $c_1,\ldots,c_m$ by counting pairs of candidates on which they disagree:
$$\kdtau(r, p) \triangleq | \{ (c_i,c_j) \mid c_i \succ_r c_j \hbox{ and } c_j \succ_p c_i \} | .$$
\end{definition}
\noindent This metric $\kdtau$ on ballots can be extended to a distance function between a ranking $r$ and a profile $\mathcal{P}$:
$$\kdtau(r, \mathcal{P}) \triangleq \sum_{p\in \mathcal{P}} \kdtau(r, p).$$

\begin{definition}[Kemeny median]
For a given profile $\mathcal{P}$, the Kemeny median is the ranking $r$ which minimizes $\kdtau(r,\mathcal{P}).$
\end{definition}
The Kemeny median satisfies many nice properties and to some extent guarantees that the chosen ranking is ``fair''. The most prominent quality is probably \textit{monotonicity}:  if voters increase a candidate's preference level, the ranking result either does not change or the promoted choice increases in overall popularity. This quality makes the median solution more robust to Byzantine behavior. The Kemeny rule is also a so called Condorcet method because if there exists a Condorcet winner, i.e., a candidate that wins all pairwise majority comparisons, it will always be ranked as the most popular choice. Besides, it only depends on the number of voters who prefer one alternative over the other and it is reinforcing, meaning that winners which were chosen independently by two different sets of voters will also become winners if the two groups are joined.


\subsection{Byzantine Setting}\label{sec:approx}

The Kemeny median cannot be computed exactly in the presence of Byzantine nodes since they might suggest rankings which have a large distance to the Kemeny median of the correct nodes thus moving the median preference away from the actual median. A notion for approximate median rankings is therefore introduced as follows:

\begin{definition}[$\alpha$-approximation of Kemeny median]
Let $m$ be a Kemeny median of a preference profile $\mathcal{P}$. An $\alpha$-approximation of $m$ is a preference ranking $m_\alpha$ satisfying $$\kdtau(m_\alpha,\mathcal{P})\leq \alpha\cdot \kdtau(m,\mathcal{P})$$
\end{definition}

As an example consider binary agreement ($m=2$): Here $\kdtau$ counts the number of correct nodes who disagree with the consensus value. Any binary Byzantine agreement algorithm that satisfies All-\hspace{0.05em}Same\hspace{0.1em}-\?Validity will also satisfy $\alpha < n-t-1$.   

Unlike binary agreement, it is not straightforward to see what a Byzantine node would choose as its ranking when the Kemeny rule determines the consensus ranking. Since the input vectors of nodes are rankings, each voter has to propose a strict order between candidates and the corresponding preference relation is transitive. A possible strategy for the Byzantine nodes would then be to choose exactly the opposite ranking of the Kemeny median of all correct nodes. We show in Appendix \ref{app:oppositeRanking} that this strategy works, but such a solution is not unique for most preference profiles. It is therefore difficult for the correct nodes to find out which of the rankings might have been Byzantine. 

\subsection{Lower Bounds on the Approximation Ratio}
\label{sec:kemenyLB}
In this section we discuss preference profiles that are vulnerable to Byzantine nodes. The first case is based on the reduction of rankings to binary agreement and gives the highest approximation ratio for $t<n/3$. 
Binary agreement does however assume that there are two groups of voters who completely disagree in their preferences. This is somewhat unlikely in practice when $m$ is sufficiently large. 
 In the second case we therefore exclude such binary instances and provide a lower bound based on Condorcet cycles within a preference profile which converges to the same value for large $m$. The approximation ratio usually depends on the ratio $n/t$, which will be denoted $k$ for the sake of simplicity. 

For our analysis, we represent the preference profile $\mathcal{P}$ as a weighted \emph{tournament graph}, i.e., a graph where the nodes represent the candidates and weighted edges represent how many voters prefer one candidate to the other. The sum of the forward and the backward edge should be equal to the total number of rankings in the corresponding preference profile.  
The ranking of a node is a directed Hamiltonian path following the order of the ranking, and all other edges are derived from the transitivity. 
For any two candidates we call the edge between these candidates a \emph{majority edge} if its backward edge has a smaller weight. The backward edge we then call a \emph{minority edge}. A Kemeny median of a weighted tournament graph is the ranking that minimizes the sum of the weights of all backward edges of the graph. 
Note that rankings restrict the power of Byzantine nodes in the sense that Byzantine nodes can only send transitive tournament graphs where every edge has weight $1$. Using tournament graphs, we can derive a lower bound for the binary case:

\begin{theorem}\label{thm:noCycles}
There is a tournament graph corresponding to a preference profile for which the Byzantine nodes may change the edge weights such that the median of the resulting preference profile is a $\frac{k}{k-2}$-approximation of the optimal median, where $k = n/t$. For $t$ close to $n/3$, this gives a $3$-approximation. 
\end{theorem}
\begin{proof}
This tournament graph is equivalent to binary agreement. Consider therefore one pair of candidates: $t$ Byzantine nodes are only able to change the median, i.e., the majority edge, between two candidates if they can swap the majority and minority edge by supporting the minority edge with their ranking.
Assume the worst case, where the forward and the backward edge both have the same weight $n/2$ after the Byzantine nodes have added their preferences. 
In the worst case the tournament graph of correct nodes had the weight $n/2$ for the majority edge. Since the correct nodes will not be able to determine the actual majority edge, they might agree on a minority edge with weight $n/2-t$ instead. The corresponding approximation ratio is then $\frac{n/2}{n/2-t} = \frac{k}{k-2}$.

Figure \ref{fig:worstCaseExample} shows a simple generalization of this argument to $m$ candidates and  proves that the lower bound of $\frac{k}{k-2}$ holds, for all $m$.
\end{proof}

\def \binaryScale  {0.45}
\begin{figure}
\begin{center}
\begin{tikzpicture}[baseline={(0,0)},scale=\binaryScale]
    \node[state,minimum size=-0pt] (C1) at (0,3.8)   {$c_1$};
    \node[state,minimum size=0pt] (C3) at (0,1)  {$c_2$};
    \node[state,minimum size=0pt] (Cm) at (0,-3)  {$c_m$};
    \node[state, fill=black!50, scale=0.1] at (0,-0.5) {};
    \node[state, fill=black!50, scale=0.1] at (0,-1) {};
    \node[state, fill=black!50, scale=0.1] at (0,-1.5) {};
   
    \path (C1) edge  [->,thick,>=latex, bend left, blue]  node [right] {} (C3)
               edge  [->,thick,>=latex, bend left, blue]  node [right] {$\frac{n}{2}$} (Cm)
          (C3) edge  [->,thick,>=latex, bend left, blue]  node [right] {} (Cm);
         
    \path (Cm) edge  [->,thick,>=latex, bend left, red]  node [left] {} (C3)
               edge  [->,thick,>=latex, bend left, red]  node [left] {$\frac{n}{2}-t$} (C1)
          (C3) edge  [->,thick,>=latex, bend left, red]  node [left] {} (C1);
       
  \end{tikzpicture}
  ~\textbf{+}~
  \begin{tikzpicture}[baseline={(0,0)},scale=\binaryScale]

    \node[state,minimum size=0pt] (F1) at (0,3.8)    {$c_1$};
    \node[state,minimum size=0pt] (F2) at (0,1)   {$c_2$};
    \node[state,minimum size=0pt] (Fm) at (0,-3)  {$c_m$};
    \node[state, fill=black!50, scale=0.1] at (0,-0.5) {};
    \node[state, fill=black!50, scale=0.1] at (0,-1) {};
    \node[state, fill=black!50, scale=0.1] at (0,-1.5) {};
 
    \path (Fm) edge  [->,thick,>=latex, bend left, red]  node [left] {} (F2)
               edge  [->,thick,>=latex, bend left, red]  node [left] {$t$} (F1)
          (F2) edge  [->,thick,>=latex, bend left, red]  node [left] {} (F1);
  
   \end{tikzpicture}
    ~\textbf{=}~
   \begin{tikzpicture}[baseline={(0,0)},scale=\binaryScale]
  
    \node[state,minimum size=0pt] (FC1) at (0,3.8)    {$c_1$};
    \node[state,minimum size=0pt] (FC2) at (0,1)   {$c_2$};
    \node[state,minimum size=0pt] (FCm) at (0,-3)  {$c_m$};
    \node[state, fill=black!50, scale=0.1] at (0,-0.5) {};
    \node[state, fill=black!50, scale=0.1] at (0,-1) {};
    \node[state, fill=black!50, scale=0.1] at (0,-1.5) {};
   
    \path (FC1) edge  [->,thick,>=latex, bend left, blue]  node [right] {} (FC2)
                edge  [->,thick,>=latex, bend left, blue]  node [right] {$\frac{n}{2}$} (FCm)
          (FC2) edge  [->,thick,>=latex, bend left, blue]  node [right] {} (FCm);
         
    \path (FCm) edge  [->,thick,>=latex, bend left, red]  node [left] {} (FC2)
                edge  [->,thick,>=latex, bend left, red]  node [left] {$\frac{n}{2}$} (FC1)
          (FC2) edge  [->,thick,>=latex, bend left, red]  node [left] {} (FC1);  
    
   \end{tikzpicture}
   ~\textbf{=}~
   \begin{tikzpicture}[baseline={(0,0)},scale=\binaryScale]
   
    \node[state,minimum size=0pt] (C11) at (0,3.8)    {$c_1$};
    \node[state,minimum size=0pt] (C12) at (0,1)   {$c_2$};
    \node[state,minimum size=0pt] (C1m) at (0,-3)  {$c_m$};
    \node[state, fill=black!50, scale=0.1] at (0,-0.5) {};
    \node[state, fill=black!50, scale=0.1] at (0,-1) {};
    \node[state, fill=black!50, scale=0.1] at (0,-1.5) {};
         
    \path (C11) edge  [->,thick,>=latex, bend left, blue]  node [right] {} (C12)
                edge  [->,thick,>=latex, bend left, blue]  node [right] {$\frac{n}{2}-t$} (C1m)
          (C12) edge  [->,thick,>=latex, bend left, blue]  node [right] {} (C1m);
         
    \path (C1m) edge  [->,thick,>=latex, bend left, red]  node [left] {} (C12)
                edge  [->,thick,>=latex, bend left, red]  node [left] {$\frac{n}{2}$} (C11)
          (C12) edge  [->,thick,>=latex, bend left, red]  node [left] {} (C11);
 
   \end{tikzpicture}
   ~\textbf{+}~
   \begin{tikzpicture}[baseline={(0,0)},scale=\binaryScale]     
   
    \node[state,minimum size=0pt] (F11) at (0,3.8)    {$c_1$};
    \node[state,minimum size=0pt] (F12) at (0,1)   {$c_2$};
    \node[state,minimum size=0pt] (F1m) at (0,-3)  {$c_m$};
    \node[state, fill=black!50, scale=0.1] at (0,-0.5) {};
    \node[state, fill=black!50, scale=0.1] at (0,-1) {};
    \node[state, fill=black!50, scale=0.1] at (0,-1.5) {};
 
    \path (F11) edge  [->,thick,>=latex, bend left, blue]  node [right] {} (F12)
                edge  [->,thick,>=latex, bend left, blue]  node [right] {$t$} (F1m)
          (F12) edge  [->,thick,>=latex, bend left, blue]  node [right] {} (F1m);
        
\end{tikzpicture}
\end{center}
\label{fig:worstCaseExample}
\caption{Two indistinguishable views on $m$ candidates for binary relations}
\medskip
\small
Note that the labels of the edges correspond to all edges in the same color. The left tournament graph is reached if $n/2$ nodes choose the order $c_1 \succ c_2 \succ \ldots\succ c_m$ and $n/2-t$ nodes choose $c_m \succ c_{m-1} \succ \ldots \succ c_1$. The right profile can be reached from a profile where $n/2$ nodes choose the order $c_m \succ c_{m-1} \succ \ldots \succ c_1$ and $n/2-t$ nodes choose $c_1 \succ c_2 \succ \ldots\succ c_m$. The Byzantine nodes can make all correct nodes see the tournament graph in the center by adding $t$ preference vectors $c_m \succ c_{m-1} \succ \ldots \succ c_1$ or $c_1 \succ c_2 \succ \ldots\succ c_m$ respectively.
\end{figure}
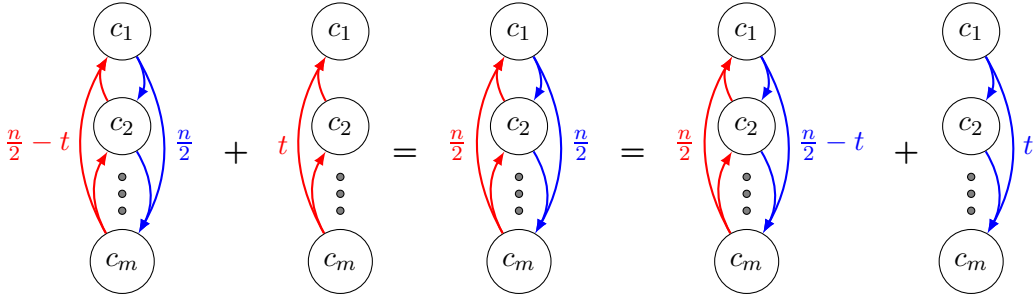

Now we present another lower bound using Condorcet cycles which can result in ambiguous views as well. Here we assume that every majority edge has a weight of more than $n/2$, thus discarding the possibility to reduce any pair of forward and backward edges in the tournament graph to binary agreement. The main difficulty in finding a good example comes from the fact that not every tournament graph has an underlying preference profile. In Appendix \ref{app:LBcycles} we discuss the necessary properties that preference profiles induce on tournament graphs and show how the best lower bound can be derived. The next theorem present the best lower bound for cycles and its generalization to $m$ candidates.

\begin{theorem} \label{thm:cycles}
By modifying directed majority cycles in the tournament graph, Byzantine nodes can increase the approximation ratio by a factor of at most $5/4$. 
\end{theorem}
\begin{proof}
We start by considering a tournament graph formed by one directed cycle of candidates $c_1,\ c_2,\ c_3$, i.e., one directed cycle formed by majority edges. 
Assume all correct nodes receive a view where $n-2t$ nodes prefer $c_1$ to $c_2$, i.e., $(c_1, c_2)$ is a majority edge; $n/2 + t$ nodes prefer $c_2$ to $c_3$ and $n/2+t$ nodes prefer $c_3$ to $c_1$. 
For $n>8f$ or equivalently $k>8$, the edge $(c_1,c_2)$ is in the median ranking of all nodes. Since the edges $(c_2,c_3)$ and $(c_3,c_1)$ cannot be both in the median ranking, the nodes have to decide for one of the rankings. In the worst case, one of these two edges was supported by all $t$ Byzantine nodes while the other edge was not supported by any Byzantine node. This leads to two views which are not distinguishable for the correct nodes, as shown in Figure \ref{fig:worstCaseExampleCircle}.
The approximation ratio for these views is $\frac{2t+n/2-t+n/2+t}{2t+n/2-2t+n/2}= \frac{n+2t}{n} = \frac{k+2}{k} < \frac{5}{4}$.

An extension to $m$ candidates gives an approximation ratio of
$\frac{2t + (m-2)\cdot(n/2-t) + (m-2)\cdot(n/2+t)}{2t + (m-2)\cdot n/2 + (m-2)\cdot(n/2-2t)} = \frac{2t+(m-2)\cdot n}{2t+(m-2)\cdot(n-2t)} = \frac{2+(m-2)k}{2+(m-2)\cdot(k-2)} \approx \frac{k}{k-2}$ for large $m$.
\end{proof}

\def \triangleScale {0.45}
\begin{figure}
\begin{center}
\begin{tikzpicture}[baseline={(0,2)},scale=\triangleScale]
    \node[state,minimum size=0pt] (C1) at (4,-2)   {$c_1$};
    \node[state,minimum size=0pt] (C2) at (-4,-2)  {$c_2$};
    \node[state,minimum size=0pt] (C3) at (0,2)      {$c_3$};
    \node[state,minimum size=0pt] (C4) at (0,5.2)      {$c_4$};
    \node[state,minimum size=0pt] (Cm) at (0,8.4)     {$c_m$};
    \node[state, shape=ellipse,minimum height =3.9cm, minimum width=3.1cm ] (x) at (-0.05,5.2) {};
    \node[state, fill=black!50, scale=0.1] at (0,6.4) {};
    \node[state, fill=black!50, scale=0.1] at (0,6.8) {};
    \node[state, fill=black!50, scale=0.1] at (0,7.2) {};

    \path (C1) edge  [->,thick,>=latex, bend left]  node [below]  {$n-3t$}   (C2)
          (C2.west) edge  [->,thick,>=latex, bend left]  node[left] {$\frac{n}{2}+t$}  (x.west)
          (x.east) edge  [->,thick,>=latex, bend left]  node [right]  {$\frac{n}{2}$}    (C1.east)
          (C1) edge  [->,thick,>=latex, bend right]  node [left,yshift=-6pt]  {$\frac{n}{2}-t$}  (x)
          (x) edge  [->,thick,>=latex, bend right]  node [right, xshift= -1pt,yshift=-6pt] {$\frac{n}{2}-2t$} (C2)
          (C2) edge  [->,thick,>=latex, bend left]  node [below]   {$2t$}     (C1)
         
          (Cm) edge  [->,thick,>=latex, bend left]  node [right] {$n-t$} (C4)
          (Cm) edge  [->,thick,>=latex, bend right]  node [left] {$n-t$} (C3)
          (C4) edge  [->,thick,>=latex, bend left]  node [right] {$n-t$} (C3);
\end{tikzpicture}
~$\mathbf{\longleftrightarrow}$~
\begin{tikzpicture}[baseline={(0,2)},scale=\triangleScale]
    \node[state,minimum size=0pt] (C1) at (4,-2)   {$c_1$};
    \node[state,minimum size=0pt] (C2) at (-4,-2)  {$c_2$};
    \node[state,minimum size=0pt] (C3) at (0,2)      {$c_3$};
    \node[state,minimum size=0pt] (C4) at (0,5.2)      {$c_4$};
    \node[state,minimum size=0pt] (Cm) at (0,8.4)     {$c_m$};
    \node[state, shape=ellipse,minimum height =3.9cm, minimum width=3.1cm ] (x) at (-0.05,5.2) {};
    \node[state, fill=black!50, scale=0.1] at (0,6.4) {};
    \node[state, fill=black!50, scale=0.1] at (0,6.8) {};
    \node[state, fill=black!50, scale=0.1] at (0,7.2) {};

    \path (C1) edge  [->,thick,>=latex, bend left]  node [below]  {$n-3t$}   (C2)
          (C2.west) edge  [->,thick,>=latex, bend left]  node[left]  {$\frac{n}{2}$}  (x.west)
          (x.east) edge  [->,thick,>=latex, bend left]  node [right]  {$\frac{n}{2}+t$}    (C1.east)
          (C1) edge  [->,thick,>=latex, bend right]  node [left,yshift=-6pt]  {$\frac{n}{2}-2t$}  (x)
          (x) edge  [->,thick,>=latex, bend right]  node [right,xshift=-1pt,yshift=-6pt]  {$\frac{n}{2}-t$} (C2)
          (C2) edge  [->,thick,>=latex, bend left]  node [below]   {$2t$}     (C1)
         
          (Cm) edge  [->,thick,>=latex, bend left]  node [right] {$n-t$} (C4)
          (Cm) edge  [->,thick,>=latex, bend right]  node [left] {$n-t$} (C3)
          (C4) edge  [->,thick,>=latex, bend left]  node [right] {$n-t$} (C3);
\end{tikzpicture}
\end{center}
\label{fig:worstCaseExampleCircle}
\caption{Two indistinguishable views on $m$ candidates for directed cycles}
\medskip
\small
We have two views which show the profiles of correct nodes only. The left tournament graph results from a profile where $n/2-t$ nodes choose $c_1 \succ c_2 \succ c_m \succ \ldots \succ c_3$, $n/2-2t$ nodes choose $c_m \succ \ldots \succ c_3 \succ c_1 \succ c_2$ and $2t$ nodes choose $c_2 \succ c_m \succ \ldots \succ c_3 \succ c_1$. 
The right tournament graph results from $n/2-2t$ nodes choosing $c_1 \succ c_2 \succ c_m \succ \ldots \succ c_3$, $n/2-t$ nodes choosing $c_m \succ \ldots \succ c_3 \succ c_1 \succ c_2$ and $2t$ nodes choosing $c_2 \succ c_m \succ \ldots \succ c_3 \succ c_1$. 
If the Byzantine nodes add $t$ profiles $c_m \succ \ldots \succ c_3 \succ c_1 \succ c_2$ to the left view, and $t$ profiles $c_1 \succ c_2 \succ c_m \succ \ldots \succ c_3$ to the right view, the resulting profiles become indistinguishable to the correct nodes. 
\end{figure}
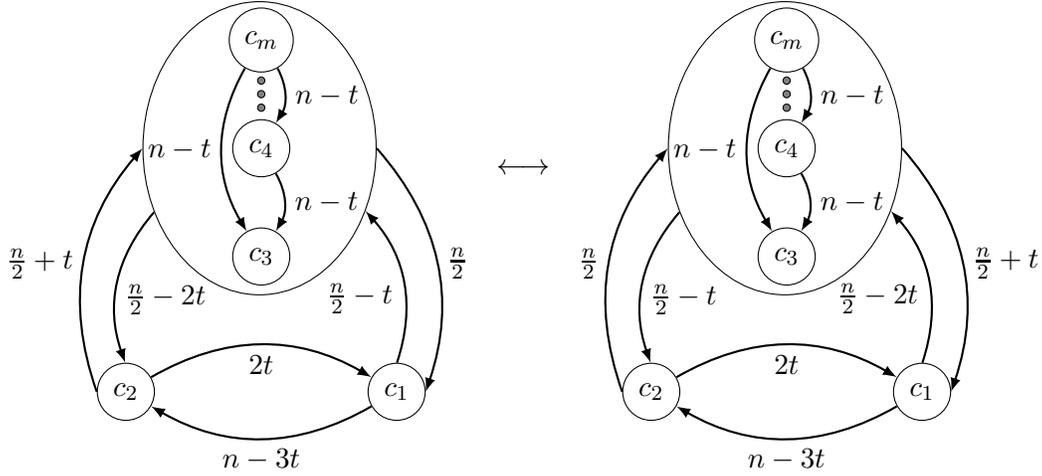

The received approximation ratio converges to the same approximation ratio as in the binary case for large $m$.

\subsection{Algorithm for Kemeny Median Approximation}\label{sec:kemenyAlgo}

In this section we present a synchronous algorithm for computing a consensus median which matches the lower bound on the approximation ratio presented in the previous section.
A simple idea is to use interactive consistency \cite{BrachaRB,TouegRB}: For $t+1$ rounds, the nodes exchange all information they have received until that round and after the $(t+1)$-st round they compute the Kemeny median from a set of rankings which they have received often enough. This algorithm guarantees that the set of rankings will be the same for each node and therefore all nodes will decide on the same ranking. The main drawback of interactive consistency is that it has a large message complexity. Since each of the correct rankings will be forwarded by each of the correct nodes, this message complexity is in $\Theta(mn^t)$.

Instead of exchanging large amounts of information, we can directly exploit the fact that the Byzantine nodes cannot change a Kemeny median of the preference profile of the correct nodes by more than a transitive tournament graph with edge weights $t$. The corresponding strategy is presented in Algorithm \ref{alg:kemenyMedianOpt}.

\begin{algorithm}
   \begin{algorithmic}[1]
   \Statex Every node $v$ executes the following algorithm
        \State broadcast own ranking $r_v$
        \State compute the Kemeny median of the received preference profile, call it $m_v$ \label{step:computeMedian}
        \State apply Algorithm \ref{alg:king} with $m_v$ as an input value \label{step:paretoOnMedian}
        
   \end{algorithmic}
   \caption{Byzantine agreement protocol for the Kemeny median (for $t < n/3$)}
   \label{alg:kemenyMedianOpt}
\end{algorithm}

The presented algorithm has the same order of round and message complexity as Algorithm \ref{alg:king}.

\begin{theorem}\label{thm:correctnessAlg2}
Algorithm \ref{alg:kemenyMedianOpt} terminates within $t+3$ rounds exchanging $O(tn^2m\log{m})$ messages. The computed consensus ranking satisfies the lower bounds from Section \ref{sec:kemenyLB} and Pareto\,-\?Validity.
\end{theorem}

In the following we give two lemmas proving the correctness of this algorithm. A full proof of Theorem \ref{thm:correctnessAlg2} is provided in Appendix \ref{app:paretoOptAlgo}.

\begin{lemma}\label{lem:localMedians}
In Step \ref{step:computeMedian} of Algorithm \ref{alg:kemenyMedianOpt}, every correct node chooses a median ranking that matches the bounds from Section \ref{sec:kemenyLB}.
\end{lemma}
\begin{proof}
Instead of all nodes in the previous section, we can consider that the Byzantine nodes change just one node's view. Since the number of Byzantine nodes remains the same and the rankings of all correct nodes are received by every node in the synchronous communication model, the Byzantine nodes can in the worst case only reach the lower bound for any correct node, but not exceed it.
\end{proof}

\begin{lemma}\label{lem:optApprox}
The computed median ranking by Step \ref{step:paretoOnMedian} of the algorithm satisfies the approximation ratios from Section \ref{sec:kemenyLB}. 
\end{lemma}
\begin{proof}
Observe that the consensus ranking is derived from a preference profile formed by the medians $m_v$. Unless some correct node disagrees on an edge in this profile, the edge will be inside the consensus median since Algorithm \ref{alg:king} satisfies Pareto\,-\?Validity. This edge may not be inside the median ranking of all correct nodes, but the approximation ratio still satisfies the bounds due to Lemma \ref{lem:localMedians}. 

If the correct nodes disagree on an edge in the preference profile of medians, there was at least one correct node who either chose the opposite edge (binary case) for its median ranking or a different edge in a directed cycle (non-binary case).
Consider the binary case first. There, the forward and the backward edge chosen as the Kemeny median $m_v$ will both satisfy the lower bound, since there is a correct node choosing either of the cases. 
For the non-binary case we need to consider directed cycles formed by the median rankings. Every directed cycle in a tournament graph implies that there is a directed sub-cycle formed by three candidates $c_i, c_j, c_k$. The corresponding preference profile of correct median rankings must contain the three opposite rankings $c_i \succ c_j \succ c_k$, $c_j \succ c_k \succ c_i$ and $c_k \succ c_i \succ c_j$. This, however, implies that all three median rankings could have been derived from the preference profile of the rankings $r_v$ by modifying edge weights by $t$. The only case from which such a situation can result is when the forward and backward edge weights of each of the three pairs of candidates differed by at most $t$ in the preference profile of rankings $r_v$. This case, again, is equivalent to the binary case and satisfies the lower bound.  

\end{proof}

Note that the computed median can have a larger Kendall's $\tau$ distance to the preference ranking of all correct nodes than any $m_v$ has in the algorithm, since the Byzantine nodes can propose their own rankings as dictators in Algorithm \ref{alg:king}. Such a ranking would still satisfy the lower bound. 


\section{Discussion and Future Work}
In this paper we introduced a new Byzantine agreement problem which extends binary Byzantine agreement to rankings. We showed that rules for choosing a consensus ranking in voting theory fit well with requirements from Byzantine agreement. 
We further considered a special voting rule, the Kemeny median, for which we provided an optimal Byzantine agreement protocol that can tolerate up to $t<n/3$ Byzantine nodes. We do not claim to have chosen the best voting rule at this point, since such a rule simply does not exist due to impossibility results in voting theory. Instead, we think of our results as an inspiration to consider a larger pool of voting rules, such as approval voting, the Godgson's rule, and many others.

Byzantine agreement on rankings can also be of interest in several applications. Consider distributed machine learning as an example. Training data is usually collected by different (potentially Byzantine) parties and is stored in different places. When the amount of data is very large or the data is too sensitive to share, 
it is impossible to transmit the data and then train a global model. Instead, these parties could train their own models locally and then vote to predict the labels for new data points. A data point could be a picture, and every trained model outputs a ranking on the possible labels (panda $\succ$ gibbon $\succ \ldots \succ$ cat).
The different parties could then use our Byzantine preferential voting protocols to find the best ranking. 


\bibliography{literature}


\newpage
\appendix


\section{Proof of Theorem \ref{thm:kingAlgo}}\label{app:kingAlgo}

We will divide the proof of Theorem \ref{thm:kingAlgo} into two parts:

\begin{lemma}\label{lem:paretoValidity}
Algorithm \ref{alg:king} satisfies Pareto\,-\?Validity.
\end{lemma}
\begin{proof}
Assume that all correct nodes agree on a binary relation between two candidates, e.g. $c_i \succ c_j$. Then all $n-t$ correct nodes will propose $c_i \succ c_j$ in Step \ref{step:propose} and every correct node will receive at least $n-t$ propose messages for the same pair. Thus, no correct node will change its ranking until a round in which the current dictator agrees on this binary relation $c_i \succ c_j$ in Step \ref{step:adaptDictator}. This holds for any pair of binary relations.
\end{proof}

\begin{lemma}
At the end of the algorithm all nodes will have agreed on the same ranking.
\end{lemma}
\begin{proof}
We will show that all correct nodes will always adjust their ranking to the ranking of a correct dictator. For this we need to verify that every correct node will enter the if-clause in Step \ref{step:adaptDictator} if the dictator is a correct node. Note that any propose message that a correct node has received $n-t$ times in Step \ref{step:adjustRanking}, will be received at least $n-2t > t+1$ times by every other correct node who will incorporate this relation into its own ranking. This way every correct dictator will agree on all binary relations that any correct node has received in Step \ref{step:adjustRanking} of the algorithm.

At latest in a round with a correct dictator all nodes will adjust their ranking to the ranking of the dictator and, by Lemma \ref{lem:paretoValidity}, will receive the same ranking at the end of the algorithm.
\end{proof}


\section{Opposite Ranking to the Kemeny Median}\label{app:oppositeRanking}

First we show that the opposite ranking of the Kemeny median always gives the worst possible solution which a Byzantine nodes can pick. Note that the sum of the weights of all edges in a tournament graph produced by the rankings of correct nodes only is $(n-t)\cdot\binom{m}{2}$. The tournament graph formed by a ranking $r$ and the tournament graph formed by the opposite ranking $\bar{r}$ are complementary graphs with respect to the tournament graph of all correct rankings. Consider the weights of all possible rankings. Each ranking that minimizes the Kendall's $\tau$ distance must have an opposite ranking that maximizes the same distance, i.e., $\kdtau(r,\mathcal{P}) + \kdtau(\bar{r}, \mathcal{P}) = (n-t)\cdot\binom{m}{2}$. This shows that the opposite ranking of the Kemeny median is also the ranking furthest away from it and can always be chosen as a Byzantine value.

The opposite ranking is however not the only ranking that the Byzantine nodes can choose. Assume all correct nodes agree on the preference $c_i \succ c_j$ such that this pair will always belong to the Kemeny median of the correct rankings. Then, the Byzantine nodes can pick either $c_i \succ c_j$ or $c_j \succ c_i$ for their ranking, since this strategy does not have any influence on the Kemeny median of all rankings. One example for such a case is depicted in Figure \ref{fig:worstCaseExampleCircle}, where the Byzantine rankings are deliberately chosen not to be the opposite rankings of the Kemeny median of all correct nodes, although such a strategy would give equivalent results. The number of possible solutions which the Byzantine nodes can choose increases with the number of binary relations on which all correct nodes agree.


\section{Proof of the Lower Bound for Cycles}\label{app:LBcycles}

Consider a cycle on the three candidates $c_1, c_2, c_3$ with majority edges $(c_1,c_2),(c_2,c_3)$ and $(c_3,c_1)$. Let the weights on majority edges $(c_1,c_2),(c_2,c_3)$ and $(c_3,c_1)$ be $x,y,z$ and the weights on corresponding minority edges $n-t-x,n-t-y,n-t-z$ respectively. 
In order to show that a cycle with such weights is the worst case we need to consider additional properties of a tournament graph:

\begin{itemize}[noitemsep]
 \item The weights of the forward and the backward edge between any two candidates sum up to the number of rankings which is $n-t$ without counting Byzantine rankings and $n$ if we count all rankings.
 
 \item Consider the tournament graph of correct nodes. Since we exclude cases which reduce to binary agreement and due to the definition of majority edges, we can assume $x\ge y\ge z\ge n/2$. 
 \item All weights satisfy the \emph{triangle inequality} on directed edges, which states that for any triplet $i,j,k$ holds
 $$w_{i,j}+w_{j,k}\ge w_{i,k},$$ 
 where $w_{i,j}$ denotes the weight on the directed edge $(c_i,c_j)$. 
 According to the triangle inequality, it holds that
 $$n-t \le x+y+z \le 2(n-t).$$
 \item Byzantine nodes can change the tournament graph by only adding new edges which also satisfy the triangle inequality. 
\end{itemize}
Based on the above conditions, the correct ranking should be $c_1\succ c_2\succ c_3$ and the total Kendall's distance is $2(n-t)-x-y+z$. In order to see which case leads us to the largest approximation ratio, we examine all other five possible rankings after considering Byzantine rankings. Since we exclude the binary case, we are able to reduce the five cases to only two which do not reduce to binary agreement: 
$c_2\succ c_3\succ c_1$ and $c_3\succ c_1 \succ c_2$.

For the case $c_2\succ c_3\succ c_1$, the total Kendall's distance is $2(n-t)+x-y-z$. In order for this ranking to be chosen as the Kemeny median after adding Byzantine rankings, the inequality $x-z\le t$ has to be satisfied. Finding the worst case for the correct nodes is equivalent to solving the optimization problem
\begin{align*}
\max& \hspace{0.5cm} \frac{2(n-t)+x-y-z}{2(n-t)-x-y+z}\\
\hbox{s.t.}&\hspace{0.5cm} x-z\le t\\
 &\hspace{0.5cm}n-t \le x+y+z\le 2(n-t)\\
 &\hspace{0.5cm}n/2 \le z\le y\le x
\end{align*}
This optimization problem has a solution when $k = n/t \ge 4$.
\begin{itemize}
\item When $k\in[4, 6]$, the maximum ratio is obtained by letting $x = n - 2t$ and $y = z = n/2.$ In this case, the maximum ratio is $2-4/k$. 
\item When $k\in(6,8]$, the maximum ratio is obtained by letting $x = n/2 + t$, $y = n-3t$ and $z = n/2.$ In this case, the maximum ratio is $1+2/k$.
\item When $k\in (8,\infty)$, the maximum ratio is obtained by letting 
$x = y = (2n-t)/3$ and $z = (2n - 4t)/3$. In this case, the maximum ratio is $(2k-1)/(2k-4).$
\end{itemize}

For the second case $c_3\succ c_1\succ c_2$, the total Kendall's distance is $2(n-t)-x+y-z$. The condition for this ranking to be the median ranking after adding Byzantine rankings is $y-z\le t.$ Finding the worst case is then equivalent to solving 
\begin{align*}
\max& \hspace{0.5cm} \frac{2(n-t)-x+y-z}{2(n-t)-x-y+z}\\
\hbox{s.t.}&\hspace{0.5cm} y-z\le t\\
 &\hspace{0.5cm}n-t \le x+y+z\le 2(n-t)\\
 &\hspace{0.5cm}n/2 \le z\le y\le x
\end{align*}
This optimization problem also only has a solution when $k = n/t \ge 4$.
\begin{itemize}
\item When $k\in[4, 8]$, the maximum ratio is obtained by letting $x = y = 3n/4 - t, z = n/2.$ In this case, the maximum ratio is $3/2-2/k$. 
\item When $k\in (8,\infty)$, the maximum ratio is obtained by letting 
$x = n-3t, y =n/2 +t$ and $z = n/2$. In this case, the maximum ratio is $1+2/k.$
\end{itemize}


\section{Proof of Theorem \ref{thm:correctnessAlg2}}\label{app:paretoOptAlgo}

In Lemma \ref{lem:localMedians} and \ref{lem:optApprox} we proved that the computed solution in Step \ref{step:paretoOnMedian} of Algorithm \ref{alg:kemenyMedianOpt} satisfies the bounds from Section \ref{sec:kemenyLB}. In the following lemma we will show that the computed median also satisfies Pareto\,-\?Validity. 

\begin{lemma}
Algorithm \ref{alg:kemenyMedianOpt} satisfies Pareto\,-\?Validity.
\end{lemma}
\begin{proof}
Assume, all nodes agree on a preference profile $c_i \succ c_j$. 
Then, there is a directed edge between $c_i$ and $c_j$ of weight $n-t$. Such an edge always belongs to a Kemeny median, since there cannot be any directed cycle formed by correct nodes only with weights $n-2t>(n-t)/2$ on all three majority edges due to the triangle inequality. This way, all correct nodes will have the preference $c_i \succ c_j$ in their median ranking $m_v$. Since Algorithm \ref{alg:king} satisfies Pareto\,-\?Validity, the consensus median ranking will also satisfy $c_i \succ c_j$.
\end{proof}

It remains to show that Algorithm \ref{alg:kemenyMedianOpt} terminates and has the same order of message complexity as Algorithm \ref{alg:king}.

\begin{lemma}
Algorithm \ref{alg:kemenyMedianOpt} terminates within $t+3$ rounds exchanging $O(tn^2m\log{m})$ messages. 
\end{lemma}
\begin{proof}
Note that Algorithm \ref{alg:king} terminates after $t+1$ rounds. Algorithm \ref{alg:kemenyMedianOpt} has two additional rounds in which the messages are exchanged. This way, our algorithm terminates within $t+3$ rounds. The message complexity of Algorithm \ref{alg:king} is $O(tn^2m\log{m})$, since in every round, each of the $n$ nodes sends  a messages of size $m\log{m}$ to $n$ other nodes. In the additional two steps of Algorithm \ref{alg:kemenyMedianOpt} $2n^2m\log{m}$ messages are exchanged which gives the same message complexity $O(tn^2m\log{m})$ for the second algorithm.
\end{proof}

\end{document}